\documentclass[11pt]{article}
\usepackage{fullpage}
\usepackage{times}
\usepackage[titletoc]{appendix}
\usepackage{epsfig}
\usepackage{amsmath,amsthm,amsfonts,amssymb,amscd}
\newtheorem{lemma}{Lemma}
\newtheorem{theorem}[lemma]{Theorem}

\newtheorem{definition}[lemma]{Definition}

\newtheorem{conjecture}{Conjecture}

\date{}

\newcommand\PH{\bf{PH}}
\newcommand\BQP{\bf{BQP}}
\newcommand\IQP{\bf{IQP}}

\newcommand\BPP{\bf{BPP}}
\newcommand\PolyP{\bf{PP}}

\newcommand\NP{\bf{NP}}
\newcommand\co{\bf{co}}

\newcommand\MA{\bf{MA}}
\newcommand\AM{\bf{AM}}

\newcommand\Poly{\bf{P}}
\newcommand\sharpP{\bf{\# P}}

\newcommand\Per{{\bf Permanent}}

\newcommand\poly{\sl{poly}}
\newcommand\var{{\Var[Q(X)]}}

\newcommand\distone {\ensuremath{\mathcal{D}_{Q,\ell}}}
\newcommand\disttwo {\ensuremath{\mathcal{D}_{Q,k}}}
\newcommand\rangeovern {\ensuremath{n>0}}

\newcommand\orbitsize {\binom{k}{\psi(\phi(x)_1)}\binom{k}{\psi(\phi(x)_2)}...\binom{k}{\psi(\phi(x)_n)}}
\newcommand\neworbitsize {\binom{k}{\psi(y_1)}\binom{k}{\psi(y_2)}...\binom{k}{\psi(y_n)}}
\def\Var{{\rm Var}\,}
     \def\E{{\rm E}\,}

\renewcommand{\leq}{\leqslant}
\renewcommand{\geq}{\geqslant}

\title{On the Power of Quantum Fourier Sampling}
\author{Bill Fefferman\thanks{Joint Center for Quantum Information and Computer Science, University of Maryland/NIST and California Institute of Technology, supported in part by NSF CCF-1423544  and BSF grant 2010120.}
\and Chris Umans\thanks{California Institute of Technology, supported by NSF CCF-1423544  and BSF grant 2010120.}}
\begin{document}

\maketitle
\setcounter{page}{0}
\thispagestyle{empty}

\begin{abstract}


A line of work initiated by Terhal and DiVincenzo \cite{di-terhal} and Bremner, Jozsa, and Shepherd \cite{iqpsampling}, shows that restricted classes of quantum computation can efficiently sample from probability distributions that cannot be exactly sampled efficiently on a classical computer, unless the $\PH$ collapses.  Aaronson and Arkhipov \cite{boson} take this further by considering a distribution that can be sampled efficiently by linear optical quantum computation, that under two feasible conjectures, cannot even be approximately sampled classically within bounded total variation distance, unless the $\PH$ collapses.  

In this work we use Quantum Fourier Sampling to construct a class of distributions that can be sampled exactly by a quantum computer.  We then argue that these distributions cannot be approximately sampled classically, unless the $\PH$ collapses, under variants of the Aaronson-Arkhipov conjectures.

In particular, we show a general class of quantumly sampleable distributions each of which is based on an ``Efficiently Specifiable" polynomial, for which a classical approximate sampler implies an average-case approximation.  This class of polynomials contains the Permanent but also includes, for example, the Hamiltonian Cycle polynomial, as well as many other familiar $\sharpP$-hard polynomials.

Since our distribution likely requires the full power of universal quantum computation, while the Aaronson-Arkhipov distribution uses only linear optical quantum computation with noninteracting bosons, why is this result interesting?  We can think of at least three reasons:
\begin{enumerate}
\item{Since the conjectures required in \cite{boson} have not yet been proven, it seems worthwhile to weaken them as much as possible.  We do this in two ways, by weakening both conjectures to apply to any ``Efficiently Specifiable" polynomial, and by weakening the so-called Anti-Concentration conjecture so that it need only hold for one distribution in a broad class of distributions.}

\item{Our construction can be understood without any knowledge of linear optics.  While this may be a disadvantage for experimentalists, in our opinion it results in a very clean and simple exposition that may be more immediately accessible to computer scientists.}

\item{It is extremely common for quantum computations to employ ``Quantum Fourier Sampling" in the following way: first apply a classically efficient function to a uniform superposition of inputs, then apply a Quantum Fourier Transform followed by a measurement. Our distributions are obtained in exactly this way, where the classically efficient function is related to a (presumed) hard polynomial. Establishing rigorously a robust sense in which the central primitive of Quantum Fourier Sampling is classically hard seems a worthwhile goal in itself.}
\end{enumerate}
\end{abstract}

\newpage
\section{Introduction}
Nearly twenty years after the discovery of Shor's factoring algorithm \cite{shor} that caused an explosion of interest in quantum computation, the complexity theoretic classification of quantum computation remains embarrassingly unsettled.

The foundational results of Bernstein and Vazirani \cite{bv}, Adleman, DeMarrais, and Huang \cite{bqppp}, and Bennett, Bernstein, Brassard and Vazarani \cite{bbbv} laid the groundwork for quantum complexity theory by defining $\BQP$ as the class of problems solvable with a quantum computer in polynomial time, and established the upper bound, $\BQP\subseteq\PolyP$, which hasn't been improved since.

In particular, given that $\BPP\subseteq\BQP$, so quantum computers are surely no less powerful than their classical counterparts, it is natural to compare the power of efficient quantum computation to the power of efficient classical verification.  Can every problem with an efficient quantum algorithm be verified efficiently?  Likewise can every problem whose solution can be verified efficiently be solved quantumly?  In complexity theoretic terms, is $\BQP\subseteq\NP$, and is $\NP\subseteq\BQP$?  Factoring is contained in $\NP\cap\co\NP$, and so cannot be $\NP$-hard unless $\NP=\co\NP$ and the $\PH$ collapses.  Thus, while being a problem of profound practical importance, Shor's algorithm does not give evidence that $\NP\subseteq\BQP$.

Even progress towards oracle separations has been agonizingly slow.  These same works that defined $\BQP$ established an oracle for which $\NP\not\subset\BQP$ \cite{bbbv} and $\BQP\not\subset\NP$ \cite{bv}.  This last result can be improved to show an oracle relative to which $\BQP\not\subset\MA$ \cite{bv}, but even finding an oracle relative to which $\BQP\not\not\subset\AM$ is still wide open.  This is particularly troubling given that, under widely believed complexity assumptions, $\NP=\MA=\AM$ \cite{klivans}.  Thus, our failure to provide an oracle relative to which $\BQP\not\subset\AM$ indicates a massive lack of understanding of the classical power of quantum computation.

Recently, two candidate oracle problems with quantum algorithms have been proven to not be contained in the $\PH$, assuming plausible complexity theoretic conjectures \cite{Scott,mine}.\footnote{Although the ``Generalized Linial-Nisan" conjecture proposed in \cite{Scott} is now known to be false \cite{A2010}.}  These advances remain at the forefront of progress on these questions.

A line of work initiated by DiVincenzo and Terhal \cite{di-terhal}, Bremner, Jozsa and Shepherd \cite{iqpsampling}, and Aaronson and Arkhipov \cite{boson} asks whether we can provide a theoretical basis for quantum superiority by looking at {\sl distribution sampling problems}.  In particular, Aaronson and Arkhipov show a {\sl distribution} that can be sampled efficiently by a particular limited form of quantum computation, that assuming the validity of two feasible conjectures, cannot be approximately sampled classically (even by a randomized algorithm with a $\PH$ oracle), unless the $\PH$ collapses.  The equivalent result for decision problems, establishing $\BQP\not\subset\BPP$ unless the $\PH$ collapses, would be a crowning achievement in quantum complexity theory.  In addition, this research has been very popular not only with the theoretical community, but also with experimentalists who hope to perform this task, ``Boson Sampling", in their labs.  Experimentally, it seems more pressing to analyze the hardness of approximate quantum sampling, since it is unreasonable to expect that any physical realization of a quantum computer can {\sl itself} exactly sample from the quantum distribution.

Interestingly, it is also known that if we can find such a quantumly sampleable distribution for which no classical approximate sampler exists, there exists a ``search" problem that can be solved by a quantum computer that cannot be solved classically \cite{samplingandsearching}.  In a search problem we are given an input $x\in\{0,1\}^n$, and our goal is to output an element in a nonempty set, $A_x\subseteq\{0,1\}^{\poly(n)}$ with high probability.  This would be one of the strongest pieces of evidence to date that quantum computers can outperform their classical counterparts.

In this work we use the same general algorithmic framework used in many quantum algorithms, which we refer to as ``Quantum Fourier Sampling", to demonstrate the existence of a general class of distributions that can be sampled exactly by a quantum computer.  We then argue that these distributions shouldn't be able to be approximately sampled classically, unless the $\PH$ collapses.  Perhaps surprisingly, we obtain and generalize many of the same conclusions as Aaronson and Arkhipov \cite{boson} with a completely different class of distributions.

Additionally, recently, and independent of us, an exciting result by Bremner, Montanaro and Shepherd \cite{newashley} obtains similar quantum ``approximate sampling" results under related but different conjectures.  While our hardness conjectures apply to a broad class of hard ``polynomials", their distribution can be sampled by a class of commuting quantum computations known as Instantaneous Quantum Polynomial time, or $\IQP$, whereas our results likely require the full power of universal quantum computation.
\section{Overview}
\subsection{Our Goals}
We want to find a class of distributions that can be sampled quantumly that cannot be approximately sampled classically, unless the $\PH$ collapses.  A natural methodology toward showing this is to prove that the existence of a classical approximate sampler implies that a $\sharpP$-hard function can be computed in the $\PH$.  By Toda's Theorem \cite{toda}, this would imply a collapse of the $\PH$.

In this work, we demonstrate a class of distributions that can be sampled exactly on a quantum computer.  We prove that the existence of an approximate sampler for these distributions implies an approximate average case solution to an ``Efficiently Specifiable" polynomial.  An Efficiently Specifiable polynomial is informally a polynomial in which the variables in each monomial can be computed efficiently from the index of the monomial.  This includes, among others, the Permanent and Hamiltonian Cycle polynomial.  

Computing a multiplicative approximation to the Permanent with integer entries in the worst-case is $\sharpP$-hard, and computing the Permanent on average  is $\sharpP$-hard (see \cite{boson} for more details).  The challenge to proving our conjectures is to put these two together to prove that an average-case multiplicative approximation to the Permanent (or for that matter, any Efficiently Specifiable polynomial) is still a $\sharpP$-hard problem.  Since we can't prove these conjectures, and we don't know the ingredients such a proof will require, it seems worthwhile to attempt to generalize the class of distributions that can be sampled quantumly.

\subsection{Our Results}
In Section \ref{algorithm} we define a general class of distributions that can be sampled exactly on a quantum computer.  The probabilities in these distributions are proportional to each different $\{\pm 1\}^n$ evaluation of a particular {\sl Efficiently Specifiable} polynomial (see Definition \ref{efficientlyspecifiable}) with $n$ variables.  We then show in Section \ref{sec-classical-nonsquash} that the existence of an approximate classical sampler for these distributions implies the existence of an {\sl additive approximate average-case} solution to the Efficiently Specifiable polynomial.  We generalize this in Section \ref{sec-non-squash} to prove that quantum computers can sample from a class of distributions in which each probability is proportional to polynomially bounded integer evaluations of an Efficiently Specifiable polynomial.

We then attempt to extend this result to quantumly sample from a distribution with probabilities proportional to exponentially bounded integer evaluations of Efficiently Specifiable polynomials.  To do this, in Section \ref{sec-squashquantum}, we introduce a variant of the Quantum Fourier Transform which we call the ``Squashed QFT".  We explicitly construct this unitary operator, and show how to use it in our quantum sampling framework.  We leave as an open question whether this unitary can be realized by an efficient quantum circuit.  We then prove in Section \ref{sec-classicalsquash}, using a similar argument to Section \ref{sec-classical-nonsquash}, that if we had a classical approximate sampler for this distribution we'd have an {\sl additive approximate average-case} solution to the Efficiently Specifiable polynomial with respect to the binomial distribution over exponentially bounded integers.

In Section \ref{sec-summary} we conclude with conjectures needed to establish the intractability of approximate classical sampling from any of our quantumly sampleable distributions.  As shown in Sections \ref{sec-classical-nonsquash} and \ref{sec-non-squash} it suffices to prove that an {\sl additive approximate average-case solution} to any Efficiently Specifiable polynomial is $\sharpP$-hard, and we conjecture that this is possible.  We also propose an ``Anti-concentration conjecture" relative to an Efficiently Specifiable polynomial over the binomial distribution, which allows us to reduce the hardness of a {\sl multiplicative approximate average-case} solution to an {\sl additive approximate average-case solution}.  Assuming this second conjecture, we can then base our first conjecture around the hardness of {\sl multiplicative}, rather than {\sl additive approximate average-case solutions} to an Efficiently Specifiable polynomial.

Our conjectures generalize conjectures in Aaronson and Arkhipov's results \cite{boson}.  They conjecture that an {\sl additive approximate average-case solution} to the Permanent with respect to the Gaussian distribution with mean $0$ and variance $1$ is $\sharpP$-hard.  They further propose an ``Anti-concentration" conjecture which allows them to reduce the hardness of {\sl multiplicative approximate average-case solutions} to the Permanent over the Gaussian distribution to the hardness of {\sl additive average case solutions} to the Permanent over the Gaussian distribution.  The parameters of our conjectures match the parameters of theirs, but our conjectures are broader, applying to any Efficiently Specifiable polynomial, a class which includes the Permanent, and a wider class of distributions, and thus is formally easier to prove.

\section{Quantum Preliminaries}
\label{sec:preliminaries}

In this section we cover the basic priciples of quantum computing needed to understand the content in the paper.  For a much more complete overview there are many references available, e.g., \cite{kitaev,nielsen}.  

The state of an {\sl $n$-qubit quantum system} is described by a unit vector in $\mathcal{H}=(\mathbb{C}^{2})^{\otimes n}$, a $2^n$-dimensional complex Hilbert space.  As per the literature we will denote the standard orthogonal basis vectors of $\mathcal{H}$ by $\{|v\rangle\}$ for $v\in\{0,1\}^n$.      

In accordance with the laws of quantum mechanics, transformations of states are described by unitary transformations acting on $\mathcal{H}$, where a {\sl unitary transformation} over $\mathcal{H}$ is a linear transformation specified by a $2^n \: \times \: 2^n$ square complex matrix $U$, such that $UU^{*} = I$, where $U^{*}$ is the conjugate transpose. Equivalently, the rows (and columns) of $U$ form an orthonormal basis.  A {\em local} unitary is a unitary that operates only on $b = O(1)$ qubits; i.e. after a suitable renaming of the standard basis by reordering qubits, it is the matrix $U \otimes I_{2^{n-b}}$, where $U$ is a $2^b \times 2^b$ unitary $U$. A local unitary can be applied in a single step of a Quantum Computer. A {\em local decomposition} of a unitary is a factorization into local unitaries. We say a $2^n \times 2^n$ unitary is {\em efficiently quantumly computable} if this factorization has at most $\poly(n)$ factors.  



 
We will need the concept of quantum evaluation of an efficiently classically computable function $f:\{0,1\}^n \rightarrow \{0,1\}^m$, which in one quantum query to $f$ maps:
$$\sum\limits_{x\in\{0,1\}^n}|x\rangle|z\rangle \rightarrow \sum\limits_{x\in\{0,1\}^n}|x\rangle|z \oplus f(x)\rangle$$  

Note that this is a unitary map, as applying it again inverts the procedure, and can be done efficiently as long as $f$ is efficiently computable.

Assuming $f$ is $\{0,1\}$-valued, we can use this state together with a simple phase flip unitary gate to prepare:
$$\sum\limits_{x\in\{0,1\}^n}\left( -1 \right)^{f(x)}|x\rangle|f(x)\rangle$$
And one more quantum query to $f$, which ``uncomputes" it, and allows us to obtain the state $\sum\limits_{x\in\{0,1\}^n} (-1)^{f(x)}|x\rangle$.
 Equivalently, if the efficiently computable function is $f:\{0,1\}\rightarrow\{\pm 1\}$ we can think of this as a procedure to prepare:

$$\sum\limits_{x\in\{0,1\}^n}f(x)|x\rangle$$ 

With two quantum queries to the function $f$. 

We close this section with an additional lemma needed for our quantum sampler.

\begin{lemma}\label{uniformfunction} Let $h:[m]\rightarrow \{0,1\}^n$ be an efficiently computable one-to-one function, 
and suppose its inverse can also be efficiently computed.
Then the superposition $\frac{1}{\sqrt{m}}\sum\limits_{x\in[m]} |h(x)\rangle$ can be efficiently prepared by a
 quantum algorithm.\end{lemma}
 \begin{proof}
 Our quantum procedure with two quantum registers proceeds as follows:
 \begin{enumerate}
 \item{Prepare $ \frac{1}{\sqrt{m}}\sum\limits_{x\in{[m]}}|x\rangle |00...0\rangle$}
 \item{Query $h$ using the first register as input and the second as output: $$\frac{1}{\sqrt{m}}\sum\limits_{x\in [m]}|x\rangle |h(x)\rangle$$}
 \item{Query $h^{-1}$ using the second register as input and the first as output: $$\frac{1}{\sqrt{m}}\sum\limits_{x\in[m]}|x\oplus h^{-1}(h(x))\rangle|h(x)\rangle=\frac{1}{\sqrt{m}}\sum\limits_{x\in[m]}|00...0\rangle|h(x)\rangle$$}
 \item{Discard first register}
 \end{enumerate}
 \end{proof}

\section{Efficiently Specifiable Polynomial Sampling on a Quantum Computer}\label{algorithm}
In this section we describe a general class of distributions that can be sampled efficiently on a Quantum Computer.

\begin{definition}[Efficiently Specifiable Polynomial]\label{efficientlyspecifiable} We say a multilinear homogenous $n$-variate polynomial $Q$ with coefficients in $\{0,1\}$ and $m$ monomials is {\sl Efficiently Specifiable} via an efficiently computable, one-to-one function $h:[m]\rightarrow\{0,1\}^n$, with an efficiently computable inverse, if:
$$Q(X_1,X_2...,X_n)=\sum\limits_{z\in [m]} {X_{1}}^{h(z)_1}{X_2}^{h(z)_2}...{X_n}^{h(z)_n}$$
\end{definition}

\begin{definition}[$\distone$]Suppose $Q$ is an Efficiently Specifiable polynomial with $m$ monomials.  For fixed $Q$ and $\ell$, we define the class of distributions
 $\distone$ over $\ell$-ary strings $y\in [0,\ell-1]^{n}$ given by:

$$\Pr_{\distone}[y]=\frac{\left| Q(Z_{y})\right|^2}{\ell^{n}m}$$
 
 Where $Z_{y}\in\mathbb{T}_{\ell}^n$ is a vector of complex values encoded by the string $y$. \end{definition}
 \begin{theorem}[Quantum Sampling Theorem]\label{essampling} Given an Efficiently Specifiable polynomial, $Q$ with $n$ variables, $m$ monomials, relative to a function $h$, and $\ell \leq exp(n)$, the resulting $\distone$ can be sampled in $poly(n)$ time on a Quantum Computer.
\end{theorem}
  \begin{proof}
~\begin{enumerate}
 \item{We start in a uniform superposition  $\frac{1}{\sqrt{m}}\sum\limits_{z\in [m]}|z\rangle$.}
 \item{We then apply Lemma \ref{uniformfunction} to prepare $\frac{1}{\sqrt{m}}\sum\limits_{z\in [m]} |h(z)\rangle$.}
 \item{Apply Quantum Fourier Transform over $\mathbb{Z}^n_\ell$ to attain \\$\frac{1}{\sqrt{\ell^{n}m}}\sum\limits_{y\in[0,\ell-1]^{n}}\sum\limits_{z\in[m]} \omega_\ell^{<y,h(z)>}|y\rangle$\\}
 \end{enumerate}

 Notice that the amplitude of each $y$ basis state in the final state after Step 3 is proportional to the value of $Q(Z_{y})$, as desired.
 
 \end{proof}
 
 \section{Classical Hardness of Efficiently Specifiable Polynomial Sampling}\label{sec-classical-nonsquash}
 We are interested in demonstrating the existence of some distribution that can be sampled exactly by a uniform family of quantum circuits, that cannot be sampled approximately classically.  Approximate here means close in Total Variation distance, where we denote the Total Variation distance between two distributions $X$ and $Y$ by $\|X-Y\|$.  Thus we define the notion of a Sampler to be a classical algorithm that approximately samples from a given class of distributions:

\begin{definition}[Sampler] Let $\{D_n\}_{\rangeovern}$ be a class of distributions where each $D_n$ is distributed over $\mathbb{C}^{n}$.  Let $r(n)\in\poly(n),\epsilon(n) \in 1/poly(n)$.  We say $S$ is a $Sampler$ with respect to $\{D_n\}$ if $\|S(0^{n},x\sim U_{\{0,1\}^{r(n)}},0^{1/{\epsilon(n)}})-D_n\|\leq \epsilon(n)$ in (classical) polynomial time.    
\end{definition}

We first recall a theorem due to Stockmeyer \cite{stockmeyer} on the ability to ``approximate count" in the $\PH$.
\begin{theorem}[Stockmeyer \cite{stockmeyer}]\label{stockmeyer}
Given as input a function $f:\{0,1\}^n\rightarrow\{0,1\}^m$ and $y\in\{0,1\}^m$, there is a procedure that outputs $\alpha$ such that:
$$(1-\epsilon)\Pr_{x\sim U_{\{0,1\}^n}}[f(x)=y]\leq \alpha \leq (1+\epsilon)\Pr_{x\sim U_{\{0,1\}^n}}[f(x)=y] $$
In randomized time $\poly(n,1/\epsilon)$ with access to an $\NP$ oracle. 
\end{theorem} 

In this section we use Theorem \ref{stockmeyer}, together with the assumed existence of a Sampler for $\distone$ to obtain hardness consequences.

In particular, we show that a Sampler would imply the existence of an efficient approximation to an Efficiently Specifiable polynomial in the following two contexts:
\begin{definition}[$\epsilon-$additive $\delta$-approximate solution]Given a distribution $D$ over $\mathbb{C}^n$ and $P:\mathbb{C}^n\rightarrow\mathbb{C}$ we say $T:\mathbb{C}^n\rightarrow\mathbb{C}$ is an $\epsilon-$additive approximate $\delta-$average case solution with respect to $D$, to $P:\mathbb{C}^{n}\rightarrow\mathbb{C}$, if $\Pr_{x\sim D}[|T(x)-P(x)|\leq \epsilon]\geq 1-\delta$.
\end{definition}

\begin{definition}[$\epsilon-$multiplicative $\delta$-approximate solution]Given a distribution $D$ over $\mathbb{C}^{n}$ and a function $P:\mathbb{C}^n\rightarrow\mathbb{C}$ we say $T:\mathbb{C}^n\rightarrow\mathbb{C}$ is an $\epsilon-$multiplicative approximate $\delta-$average case solution with respect to $D$, if $\Pr_{x\sim D}[|T(x)-P(x)|\leq\epsilon |P(x)|]\geq 1-\delta$.
\end{definition}

These definitions formalize a notion that we will need, in which an efficient algorithm computes a particular hard function approximately only on most inputs, and can act arbitrarily on a small fraction of remaining inputs.

In this section, we focus on the uniform distribution on $\{\pm 1\}$ strings, and a natural generalization:
\begin{definition}[$\mathbb{T}_\ell$]Given $\ell>0$, we define the set $\mathbb{T}_{\ell}=\{\omega_\ell^{0},\omega_\ell^{1}...,\omega_\ell^{\ell-1}\}$ where $\omega_\ell$ is a primitive $\ell$-th root of unity.\end{definition}
We note that $\mathbb{T}_\ell$ is just $\ell$ evenly spaced points on the unit circle, and $\mathbb{T}_2=\{\pm 1\}$.
\begin{theorem}[Complexity consequences of Sampler]\label{additive-m}  Given an Efficiently Specifiable polynomial $Q$ with $n$ variables and $m$ monomials, and a Sampler $S$ with respect to $\distone$, there is a randomized procedure $T:\mathbb{C}^{n}\rightarrow\mathbb{C}$, an $(\epsilon\cdot m)-$additive approximate $\delta-$average case solution with respect to the uniform distribution over $\mathbb{T}_\ell^{n}$, to the $|Q|^2$ function, that runs in randomized time $\poly(n,1/\epsilon,1/\delta)$ with access to an $\NP$ oracle. 
\end{theorem}
\begin{proof}
We need to give a procedure that outputs an $\epsilon m$-additive estimate to the $|Q|^2$ function evaluated at a uniform setting of the variables, with probability $1-\delta$ over choice of setting.  Setting $\beta=\frac{\epsilon\delta}{16}$, suppose $S$ samples from a distribution $\mathcal{D'}$ such that $\|\distone-\mathcal{D}'\|\leq\beta$.  We let $p_y$ be $\Pr_{\distone}[y]$ and $q_y$ be $\Pr_{\mathcal{D}'}[y]$.\\

Our procedure picks a uniformly chosen encoding of a setting $y\in[0,\ell-1]^{n}$, and outputs an estimate $\tilde{q}_y$.  Note that $p_y=\frac{|Q(Z_{y})|^2}{\ell^{n}m}$.  Thus our goal will be to output a $\tilde{q}_y$ that approximates $p_y$ within additive error $\epsilon \frac{m}{\ell^{n}m}=\frac{\epsilon}{\ell^{n}}$, in time polynomial in $n$, $\frac{1}{\epsilon}$, and $\frac{1}{\delta}$.\\

We need:

$$\Pr_y[|\tilde{q}_y-p_y| > \frac{\epsilon}{\ell^{n}}] \leq\delta$$

First, define for each $y$, $\Delta_y=|p_y-q_y|$, and so $\|\distone-\mathcal{D}'\|=\frac{1}{2}\sum\limits_y[\Delta_y]$. \\

Note that:

$$E_y[\Delta_y]=\frac{\sum\limits_y[\Delta_y]}{\ell^{n}}=\frac{2\beta}{\ell^{n}}$$

And applying Markov's inequality, $\forall k>1$,

$$\Pr_y[\Delta_y>\frac{k2\beta}{\ell^{n}}]<\frac{1}{k}$$

Setting $k=\frac{4}{\delta}, \beta=\frac{\epsilon\delta}{16}$, we have,

$$\Pr_y[\Delta_y>\frac{\epsilon}{2}\cdot\frac{1}{\ell^{n}}]<\frac{\delta}{4}$$

Then use approximate counting (with an $\NP$ oracle), using Theorem \ref{stockmeyer} on the randomness of $S$ to obtain an output $\tilde{q}_y$ so that, for all $\gamma>0$, in time polynomial in $n$ and $\frac{1}{\gamma}$:

$$\Pr[|\tilde{q}_y-q_y| >\gamma\cdot q_y]<\frac{1}{2^{n}}$$
Because we can amplify the failure probability of Stockmeyer's algorithm to be inverse exponential.
Note that:

$$E_y[q_y]=\frac{\sum\limits_y q_y}{\ell^{n}}=\frac{1}{\ell^{n}}$$

Thus,

$$\Pr_y[q_y>\frac{k}{\ell^{n}}]<\frac{1}{k}$$

Now, setting $\gamma=\frac{\epsilon\delta}{8}$ and applying the union bound:


\begin{align*}
\Pr_y[|\tilde{q}_y-p_y|>\frac{\epsilon}{\ell^{n}}] &\leq\Pr_y[|\tilde{q}_y-q_y|>\frac{\epsilon}{2}\cdot\frac{1}{\ell^{n}}]+\Pr_y[|q_y-p_y|>\frac{\epsilon}{2}\cdot\frac{1}{\ell^{n}}]
\\  & \leq \Pr_y[q_y>\frac{k}{\ell^{n}}]+\Pr[|\tilde{q}_y-q_y|>\gamma\cdot q_y]+\Pr_y[\Delta_y>\frac{\epsilon}{2}\cdot\frac{1}{\ell^{n}}]
\\ & \leq\frac{1}{k}+\frac{1}{2^{n}}+\frac{\delta}{4}
\\ & \leq \frac{\delta}{2}+\frac{1}{2^{n}}\leq \delta.
\end{align*}

\end{proof}

Now, as will be proven in Appendix \ref{variance-section}, the variance, $\var$, of the distribution over $\mathbb{C}$ induced by an Efficiently Specifiable $Q$ with $m$ monomials, evaluated at uniformly distributed entries over $\mathbb{T}_\ell^{n}$ is $m$, and so the preceding Theorem \ref{additive-m} promised us we can achieve an $\epsilon\var$-additive approximation to $Q^2$, given a Sampler.  We now show that, under a conjecture, this approximation can be used to obtain a good multiplicative estimate to $Q^2$.  This conjecture effectively states that the Chebyshev inequality for this random variable is tight.

\begin{conjecture}[Anti-Concentration Conjecture relative to an $n$-variate polynomial $Q$ and distribution $\mathcal{D}$ over $\mathbb{C}^n$]\label{anticon}  There exists a polynomial $p$ such that for all $n$ and $\delta>0$,
$$\Pr_{X\sim \mathcal{D}}\left[\left|Q(X)\right|^2<\frac{\var}{p(n,1/\delta)}\right]<\delta$$

\end{conjecture}

\begin{theorem}\label{thm-anticon}
Assuming Conjecture \ref{anticon}, relative to an Efficiently Specifiable polynomial $Q$ and a distribution $\mathcal{D}$, an $\epsilon\var$-additive approximate $\delta$-average case solution with respect to $D$, to the $|Q|^2$ function  can be used to obtain an $\epsilon'\leq poly(n)\epsilon$-multiplicative approximate $\delta'=2\delta$-average case solution with respect to $\mathcal{D}$ to $|Q|^2$.
\end{theorem}
\begin{proof}
 
Suppose $\lambda$ is, with high probability, an $\epsilon\var$-additive approximation to $|Q(X)|^2$, as guaranteed in the statement of the Theorem.  This means:

$$\Pr_{X\sim \mathcal{D}}\left[\left|\lambda-\left|Q(X)\right|^2\right|>\epsilon \var\right]< \delta$$

Now assuming Conjecture \ref{anticon} with polynomial $p$, we will show that $\lambda$ is also a good multiplicative approximation to $|Q(X)|^2$ with high probability over $X$.\\

By the union bound,

\begin{align*}
\Pr_{X\sim \mathcal{D}}\left[\frac{\left|\lambda-\left|Q(X)\right|^2\right|}{\epsilon p(n,1/\delta)}>\left|Q(X)\right|^2\right]  \leq & \Pr_{X\sim D} \left[\left|\lambda-\left|Q(X)\right|^2\right|>\epsilon \var \right] \: +
\\ & \Pr_{X\sim \mathcal{D}}\left[\frac{\epsilon \var}{\epsilon p(n,1/\delta)}>\left|Q(X)\right|^2\right]
\\ & \leq 2\delta
\end{align*}

Where the second line comes from Conjecture \ref{anticon}.  Thus we can achieve any desired multiplicative error bounds $(\epsilon',\delta')$ by setting $\delta=\delta'/2$ and $\epsilon=\epsilon'/p(n,1/\delta)$.

\end{proof}

For the results in this section to be meaningful, we simply need the Anti-Concentration conjecture to hold for some Efficiently Specifiable polynomial that is $\sharpP$-hard to compute, relative to any distribution we can sample from (either $U_n$, or $\mathcal{B}(0,k)^n$).  We note that Aaronson and Arkhipov \cite{boson} conjectures the same statement as Conjecture \ref{anticon} for the special case of the $\Per$ function relative to matrices with entries distributed independently from the complex Gaussian distribution of mean 0 and variance 1.

Additionally, we acknowledge a result of Tao and Vu who show:
\begin{theorem}[Tao \& Vu \cite{taovu}]  For all $\epsilon>0$ and sufficiently large $n$, 
$$\Pr_{X\in \{\pm 1\}^{n \times n}}\left[\left|\Per[X]\right| < \frac{\sqrt{n!}}{n^{\epsilon n}}\right]<\frac{1}{n^{0.1}} $$

\end{theorem}

Which comes quite close to our conjecture for the case of the $\Per$ function and uniformly distributed $\{\pm 1\}^{n\times n}=\mathbb{T}_2^{n \times n}$ matrix.  More specifically, for the above purpose of relating the additive hardness to the multiplicative, we would need an upper bound of any inverse polynomial $\delta$, instead of a fixed $n^{-0.1}$.

\section{Sampling from Distributions with Probabilities Proportional to $[-k,k]$ Evaluations of Efficiently Specifiable Polynomials}\label{sec-non-squash}

In the prior sections we discussed quantum sampling from distributions in which the probabilities are proportional to evaluations of Efficiently Specifiable polynomials evaluated at points in $\mathbb{T}_{\ell}^n$.  In this section we show how to generalize this to quantum sampling from distributions in which the probabilities are proportional to evaluations of Efficiently Specifiable polynomials evaluated at polynomially bounded integer values.  In particular, we show a simple way to take an Efficiently Specifiable polynomial with $n$ variables and create another Efficiently Specifiable polynomial with $kn$ variables, in which evaluating this new polynomial at $\{-1,+1\}^{kn}$ is equivalent to evaluation of the old polynomial at $[-k,k]^n$.

\begin{definition}[$k$-valued equivalent polynomial]
For every Efficiently Specifiable polynomial $Q$ with $m$ monomials and every fixed $k>0$ consider the polynomial $Q'_k:\mathbb{T}^{kn}_2\rightarrow\mathbb{R}$ defined by replacing each variable $x_i$ in $Q$ with the sum of $k$ new variables $x_i^{(1)}+x_i^{(2)}+...+x_i^{(k)}$.  We will call $Q'_k$ the $k$-valued equivalent polynomial with respect to $Q$.
\end{definition}

Note that a uniformly chosen $\{\pm 1\}$ assignment to the variables in $Q'_k$ induces an assignment to the variables in $Q$, distributed from a distribution we call $\mathcal{B}(0,k)$:
\begin{definition}[$\mathcal{B}(0,k)$] For $k$ a positive integer, we define the distribution $\mathcal{B}(0,k)$ supported over the odd integers in the range $[-k,k]$ (if $k$ is odd), or even integers in the range $[-k,k]$ (if $k$ is even), so that:
\[\Pr_{\mathcal{B}(0,k)}[y]= \left\{
  \begin{array}{lr}
    \frac{\binom{k}{\frac{k+y}{2}}}{2^k} &  \: if \: y \: and \:k \: are \: both \: odd \:or \: both \: even\\
    0 &  otherwise
  \end{array}
\right. \]
\end{definition}
\begin{theorem}\label{kvaluesampling}  Given an Efficiently Specifiable polynomial $Q$ with $n$ variables and $m$ monomials, let $Q'_k$ be its $k$-valued equivalent polynomial.  For all $\ell<exp(n)$, we can quantumly sample from the distribution $\mathcal{D}_{Q'_k,\ell}$ in time $\poly(n,k)$.
\end{theorem}
\begin{proof}
Our proof follows from the following lemma, which proves that $Q'_k$ is Efficiently Specifiable.
\begin{lemma}\label{kvaluedes} Suppose $Q$ is an $n$-variate, homogeneous degree $d$ Efficiently Specifiable polynomial with $m$ monomials relative to a function $h:[m]\rightarrow\{0,1\}^n$.  Let $k\leq poly(n)$ and let $Q'_k$ be the $k$-valued equivalent polynomial with respect to $Q$.  Then $Q'_k$ is Efficiently Specifiable with respect to an efficiently computable function $h':[m]\times[k]^d\rightarrow\{0,1\}^{kn}$. 
\end{lemma}
\begin{proof}

We first define and prove that $h'$ is efficiently computable.  We note that if there are $m$ monomials in $Q$, there are $mk^d$ monomials in $Q'$.  As before, we'll think of the new variables in $Q'_k$ as indexed by a pair of indices, a ``top index" in $[k]$ and a ``bottom index" in $[n]$.  Equivalently we are labeling each variable in $Q'_k$ as $x_{i}^{(j)}$, the $j$-th copy of the $i$-th variable in $Q$.  We are given $x\in[m]$ and $y_1,y_2,...,y_d\in[k]$.  Then, for all $i\in[n]$ and $j\in[k]$, we define the output, $z=h'(x,y_1,y_2,...,y_d)_{i,j}=1$  iff:
\begin{enumerate}
\item{$h(x)_i=1$}
\item{If $h(x)_i$ is the $\ell\leq d$-th non-zero element of $h(x)$, then we require $y_\ell=j$}
\end{enumerate}

We will now show that $h'^{-1}$ is efficiently computable.  As before we will think of $z\in\{0,1\}^{kn}$ as being indexed by a pair of indices, a `top index" in $[k]$ and a ``bottom index" in $[n]$.  Then we compute $h'^{-1}(z)$ by first obtaining from $z$ the bottom indices $j_1,j_2,...,j_d$ and the corresponding top indices, $i_1,i_2,...,i_d$.  Then obtain from the bottom indices the string $x\in\{0,1\}^n$ corresponding to the indices of variables used in $Q$ and output the concatenation of $h^{-1}(x)$ and $j_1,j_2,...,j_d$.\\
\end{proof}
Theorem \ref{kvaluesampling} now follows from Lemma \ref{kvaluedes}, where we established that $Q'_k$ is Efficiently Specifiable, and Theorem \ref{essampling}, where we established that we can sample from $\mathcal{D}_{Q'_k,\ell}$ quantumly.
\end{proof}

\begin{theorem}
Let $\var=\Var[Q(X_1,X_2,...,X_n)]$ denote the variance of the distribution over $\mathbb{R}$ induced by $Q$ with assignments distributed from $\mathcal{B}(0,k)^n$.  Given a Sampler $S$ with respect to $\mathcal{D}_{Q'_{k},2}$, we can find a randomized procedure $T:\mathbb{R}^n\rightarrow\mathbb{R}$, an $\epsilon\var$-additive approximate $\delta$-average case solution to $Q^2$ with respect to $\mathcal{B}(0,k)^n$ that runs in time $poly(n,1/\epsilon,1/\delta)$ with access to an $\NP$ oracle. 
\end{theorem}
\begin{proof}
We begin by noting that $Q'_k$ is a polynomial of degree $d$ that has $kn$ variables and $m'=mk^d$ monomials.  By Theorem \ref{additive-m} we get that a Sampler with respect to $\mathcal{D}_{Q'_{k,2}}$ implies there exists $A$, an $\epsilon m'$-additive approximate $\delta$-average case solution to ${Q'_k}^2$ with respect to $U_{\{\pm 1\}^{kn}}$ that runs in time $poly(n,1/\epsilon,1/\delta)$ with access to an $\NP$ oracle.  We need to show the existence of an $A'$, an $\epsilon m'$-additive approximate $\delta$-average case solution to ${Q'_k}^2$ with respect to the $\mathcal{B}(0,k)^n$ distribution.  

We think of $A'$ as receiving an input, $z\in [-k,k]^n$ drawn from $\mathcal{B}(0,k)^n$.  $A'$ picks $y$ uniformly from the orbit of $z$ over $\{\pm 1\}^{kn}$ and outputs $A(y)$.  Now:

\begin{align}
\Pr_{z\sim\mathcal{B}(0,k)^n}\left[\left|A'(z)-Q^2(z)\right|\leq \epsilon m'\right]&=\Pr\limits_{z\sim \mathcal{B}(0,k)^n,y\sim_R {orbit(z)}}\left[\left|A(y)-Q^2(z)\right|\leq \epsilon m'\right]\\
&=\Pr_{y\sim U_{\lbrace\pm 1\rbrace^{kn}}}\left[\left|A(y)-Q'_k(y)\right|\leq \epsilon m'\right]\geq 1-\delta\\
\end{align}


Thus, because a uniformly chosen $\{\pm 1\}^{kn}$ assignment to the variables in $Q'_k$ induces a $\mathcal{B}(0,k)^n$ distributed assignment to the variables in $Q$, this amounts to an $\epsilon m'$-additive approximate $\delta$-average case solution to $Q^2$ with respect to $\mathcal{B}(0,k)^n$.  In Appendix \ref{variance-section} we prove that $\var$ is $m'$ as desired.
\end{proof}

\section{The ``Squashed" QFT}\label{sec-squashquantum}
In this section we begin to prove that Quantum Computers can sample efficiently from distributions with probabilities proportional to evaluations of Efficiently Specifiable polynomials at points in $[-k,k]^n$ for $k=exp(n)$.  Note that in the prior quantum algorithm of Section \ref{algorithm} we would need to invoke the QFT over $\mathbb{Z}_2^{kn}$, of dimension doubly-exponential in $n$.  Thus we need to define a new Polynomial Transform that can be obtained from the standard Quantum Fourier Transform over $\mathbb{Z}_2^n$, which we refer to as the ``Squashed QFT".  Now we describe the unitary matrix which implements the Squashed QFT. 
 
Consider the $2^k \times 2^k$ matrix $D_k$, whose columns are indexed by all possible $2^{k}$ multilinear monomials of the variables $x_1,x_2,...,x_k$ and the rows are indexed by the $2^k$ different $\{-1,+1\}$ assignments to the variables. The $(i,j)$-th entry is then defined to be the evaluation of the $j$-th monomial on the $i$-th assignment.  We note in passing that, defining $\bar{D}_k$ to be the matrix whose entries are the entries in $D_k$ normalized by $1/\sqrt{2^k}$ gives us the Quantum Fourier Transform matrix over $\mathbb{Z}_2^k$.  It is clear, by the unitarity of the Quantum Fourier Transform, that the columns (and rows) in $D_k$ are pairwise orthogonal.

Now we define the ``Elementary Symmetric Polynomials":
\begin{definition}[Elementary Symmetric Polynomials]  We define the $j$-th Elementary Symmetric Polynomial on $k$ variables for $j\in[0,k]$ to be:
$$p_j(X_1,X_2,...,X_k)=\sum\limits_{1\leq \ell_1 < \ell_2< ... < \ell_j\leq k}X_{\ell_1}X_{\ell_2}...X_{\ell_j}$$
\end{definition}

In this work we will care particularly about the first two elementary symmetric polynomials, $p_0$ and $p_1$ which are defined as $p_0(X_1,X_2,...,X_k)=1$ and $p_1(X_1,X_2,...,X_k)=\sum\limits_{1\leq \ell \leq k}X_\ell$.

Consider the $(k+1) \times (k+1)$ matrix, $\tilde{D}_k$, whose columns are indexed by elementary symmetric polynomials on $k$ variables and whose rows are indexed by equivalence classes of assignments in $\mathbb{Z}_2^k$ under $S_k$ symmetry.  We obtain $\tilde{D}_k$ from $D_{k}$ using two steps.\\

First obtain a $2^{k} \times (k+1)$ rectangular matrix $\tilde{D}^{(1)}_k$ whose rows are indexed by assignments to the variables $x_1,x_2,...,x_k\in\{\pm 1\}^k$ and columns are the entry-wise sum of the entries in each column of $D_k$ whose monomial is in each respective elementary symmetric polynomial.\\

Then obtain the final $(k+1)\times (k+1)$ matrix $\tilde{D}_k$ by taking $\tilde{D}^{(1)}_k$ and keeping only one representative row in each equivalence class of assignments under $S_k$ symmetry.  We label the equivalence classes of assignments under $S_k$ symmetry $o_0,o_1,o_2,...,o_k$ and note that for each $i\in [k]$, $|o_i|=\binom{k}{i}$.  Observe that $\tilde{D}_k$ is precisely the matrix whose $(i,j)$-th entry is the evaluation of the $j$-th symmetric polynomial evaluated on an assignment in the $i$-th symmetry class.\\

\begin{theorem}
The columns in the matrix $\tilde{D}^{(1)}_k$ are pairwise orthogonal.
\end{theorem}
\begin{proof}
Note that each column in the matrix $\tilde{D}^{(1)}_k$ is the sum of columns in $D_k$ each of which are orthogonal. We can prove this theorem by observing that if we take any two columns in $D^{(1)}_k$, called $c_1,c_2$, where $c_1$ is the sum of columns $\{u_i\}$ of $D_k$ and $c_2$ is the sum of columns $\{v_i\}$ of $D_k$.  The inner product, $\langle c_1,c_2 \rangle$ can be written:
$$\langle \sum\limits_{i} u_i,\sum\limits_{j} v_j\rangle=\sum\limits_{i,j}\langle u_i,v_j\rangle=0$$

\end{proof}

\begin{theorem}
Let $L$ be the $(k+1) \times (k+1)$ diagonal matrix with $i$-th entry equal to $\sqrt{o_i}$.  Then the columns of $L\cdot\tilde{D}_k$ are orthogonal.
\end{theorem}
\begin{proof}
Note that the value of the symmetric polynomial at each assignment in an equivalence class is the same.  We have already concluded the orthogonality of columns in $\tilde{D}^{(1)}_k$.  Therefore if we let $a$ and $b$ be any two columns in the matrix $\tilde{D}_k$, and their respective columns be $\bar{a},\bar{b}$ in $\tilde{D}^{(1)}_k$, we can see:

$$\sum\limits_{i=0}^{k}\left(a_ib_i|o_i|\right)=\sum\limits_{i=0}^{2^k}\bar{a}_i\bar{b}_i=0$$

From this we conclude that the columns of the matrix $L\cdot\tilde{D}_k$, in which the $i$-th row of $\tilde{D}_k$ is multiplied by $\sqrt{o_i}$, are orthogonal.
\end{proof}

\begin{theorem}
We have just established that the columns in the matrix $L\cdot\tilde{D}_k$ are orthogonal.  Let the $k+1 \times k+1$ diagonal matrix $R$ be such that so that the columns in $L\cdot\tilde{D}_k\cdot R$ are orthonormal, and thus $L\cdot\tilde{D}_k\cdot R$ is unitary.  Then the first two nonzero entries in $R$, which we call $r_0,r_1$, corresponding to the normalization of the column pertaining to the zero-th and first elementary symmetric polynomial, are $1/\sqrt{2^k}$ and $\frac{1}{\sqrt{\sum\limits_{i=0}^k\left[\binom{k}{i}\left(k-2i\right)^2\right]}}$. 
\end{theorem}

\begin{proof}
First we calculate $r_0$.  Since we wish for a unitary matrix, we want the $\ell_2$ norm of the first column of $\tilde{D}_k$ to be 1, and so need:
$$r_0^2\sum\limits_{i=0}^k\left(\sqrt{o_i}\right)^2=r_0^2\sum\limits_{i=0}^k\binom{k}{i}=1$$

And so $r_0$ is $1/\sqrt{2^k}$ as desired.\\

Now we calculate $r_1$, the normalization in the column of $\tilde{D}_k$ corresponding to the first elementary symmetric polynomial.  Note that in i-th equivalence class of assignments we have exactly $i$ negative ones and $k-i$ positive ones.  Thus the value of the first symmetric polynomial is the sum of these values, which for the $i-th$ equivalence class is precisely $k-2i$.  Then we note the normalization in each row is $\sqrt{\binom{k}{i}}$.  Thus we have
$$r_1^2\sum\limits_{i=0}^k{\left[\sqrt{\binom{k}{i}}\left(k-2i\right)\right]^2}=1$$
Thus $r_1=\frac{1}{\sqrt{\sum\limits_{i=0}^k\left[\binom{k}{i}\left(k-2i\right)^2\right]}}$ as desired.

\end{proof}


\section{Using our ``Squashed QFT" to Quantumly Sample from Distributions of Efficiently Specifiable Polynomial Evaluations}\label{sec-quantumsquash}

In this section we use the unitary matrix developed earlier to quantumly sample distributions with probabilities proportional to evaluations of Efficiently Specifiable polynomials at points in $[-k,k]^n$ for $k=exp(n)$.  Here we assume that we have an efficient quantum circuit decomposition for this unitary.  The prospects for this efficient decomposition are discussed in Section \ref{sec-summary}.

For convenience, we'll define a map $\psi:[-k,k]\rightarrow [0,k]$, for $k$ even, with 
\[\psi(y)= \left\{
  \begin{array}{lr}
    \frac{k+y}{2} &  \: if \: y \: is \: even\\
    0 &  otherwise
  \end{array}
\right. \]

\begin{definition}
Suppose $Q$ is an Efficiently Specifiable polynomial $Q$ with $n$ variables and $m$ monomials, and, for $k\leq exp(n)$, let $Q'_{k}$ be its $k$-valued equivalent polynomial.  Let $\var$ be the variance of the distribution over $\mathbb{R}$ induced by $Q$ with assignments to the variables distributed over $\mathcal{B}(0,k)^n$ (or equivalently, we can talk about $\Var[Q'_k]$ where each variable in $Q'_k$ is independently uniformly chosen from $\{\pm 1\}$), as calculated in Appendix \ref{variance-section}.  Then we define the of distribution $\disttwo $ over $n$ tuples of integers in $[-k,k]$ by:
$$\Pr_{\disttwo}[y]=\frac{Q(y)^2{{\neworbitsize}}}{2^{kn}\var}$$
\end{definition} 
\begin{theorem}
By applying $(L\cdot\tilde{D}_k\cdot R)^{\otimes n}$ in place of the Quantum Fourier Transform over $\mathbb{Z}_2^n$ in Section \ref{algorithm} we can efficiently quantumly sample from $\disttwo$.
\end{theorem}
\begin{proof}
Since we are assuming $Q$ is Efficiently Specifiable, let $h:[m]\rightarrow \{0,1\}^{n}$ be the invertible function describing the variables in each monomial.  We start by producing the state over $k+1$ dimensional qudits: $$\frac{1}{\sqrt{m}}\sum\limits_{z\in [m]} |h(z)\rangle$$ Which we prepare via the procedure described in Lemma \ref{uniformfunction}.\\

Instead of thinking of $h$ as mapping an index of a monomial from $[m]$ to the variables in that monomial, we now think of $h$ as taking an index of a monomial in $Q$ to a polynomial expressed in the $\{1,x^{(1)}+x^{(2)}+...+x^{(k)}\}^n$ basis.\\

Now take this state and apply the unitary (which we assume can be realized by an efficient quantum circuit) $(L\cdot\tilde{D}_k\cdot R)^{\otimes n}$.\\  

Notice each $y\in[-k,k]^{n}$ has an associated amplitude: $$\alpha_y=\frac{r_0^{n-d}r_1^{d}Q(y)\sqrt{\neworbitsize}}{\sqrt{m}}$$
Letting $p_y=\Pr_{\disttwo}[y]$, note that, by plugging in $r_0,r_1$ from Section \ref{sec-squashquantum}:
\begin{align*}
\alpha_y^2 &=\frac{Q(y)^2{\neworbitsize}r_0^{2(n-d)}r_1^{2d}}{m}\\
&=\frac{Q(y)^2{\neworbitsize}}{{m{2^{k(n-d)}}\left(\sum\limits_{i=0}^k\left[\binom{k}{i}\left(k-2i\right)^2\right]\right)^d}}\\&=\frac{Q(y)^2{{\neworbitsize}}}{2^{kn-kd}\var 2^{kd}}=\frac{Q(y)^2{{\neworbitsize}}}{2^{kn}\var}=p_y
\end{align*}
\end{proof}

\section{The Hardness of Classical Sampling from the Squashed Distribution}\label{sec-classicalsquash}
In this section, as before, we use Stockmeyer's Theorem (Theorem \ref{stockmeyer}), together with the assumed existence of a Sampler for $\disttwo$ to obtain hardness consequences for classical sampling with $k\leq exp(n)$.
\begin{theorem}  Given an Efficiently Specifiable polynomial $Q$ with $n$ variables and $m$ monomials, let $Q'_k$ be its $k$-valued equivalent polynomial, for some fixed $k\leq exp(n)$.  Suppose we have a Sampler $S$ with respect to our quantumly sampled distribution class, $\disttwo$, and let $\var$ denote the variance of the distribution over $\mathbb{R}$ induced by $Q$ with assignments distributed from $\mathcal{B}(0,k)^n$.  Then we can find a randomized procedure $T:\mathbb{R}^n\rightarrow\mathbb{R}$, an $\epsilon\var$-additive approximate $\delta$-average case solution to $Q^2$ with respect to $\mathcal{B}(0,k)^n$ that runs in time $poly(n,1/\epsilon,1/\delta)$ with access to an $\NP$ oracle. 
\end{theorem}
\begin{proof}
Setting $\beta=\epsilon\delta/16$, suppose $S$ samples from a class of distributions $\mathcal{D}'$ so that $\|\disttwo-\mathcal{D}'\|\leq\beta$.  Let $q_{y}=\Pr_{\mathcal{D}'}[y]$.\\

We define $\phi:\{\pm 1\}^{kn}\rightarrow [-k,k]^n$ to be the map from each $\{\pm 1\}^{kn}$ assignment to its equivalence class of assignments, which is $n$ blocks of even integral values in the interval $[-k,k]$.  Note that, given a uniformly random $\{\pm 1\}^{kn}$ assignment, $\phi$ induces the $\mathcal{B}(0,k)$ distribution over $[-k,k]^n$.\\

Our procedure picks a $y\in[-k,k]^{n}$ distributed\footnote{We can do this when $k=exp(n)$ by approximately sampling from the Normal distribution, with only $\poly(n)$ bits of randomness, and using this to approximate $\mathcal{B}(0,k)$ to within additive error 1/$\poly(n)$ e.g., \cite{boxmuller,normal}.} via $\mathcal{B}(0,k)^n$, and outputs an estimate $\tilde{q}_{y}$.  Equivalently, we analyze this procedure by considering a uniformly distributed $x\in\{\pm 1\}^{kn}$ and then returning an approximate count, $\tilde{q}_{\phi(x)}$ to $q_{\phi(x)}$.  We prove that our procedure runs in time $poly(n,1/\epsilon,1/\delta)$ with the guarantee that:

$$\Pr_x\left[\frac{|\tilde{q}_{\phi(x)}-p_{\phi(x)}|}{\orbitsize} > \frac{\epsilon}{2^{kn}}\right] \leq\delta$$\\

And by our above analysis of the quantum sampler:
 $$p_{\phi(x)}=\frac{Q(\phi(x))^2\orbitsize}{2^{kn}\var}$$\\

Note that: $\frac{1}{2}\sum\limits_{y\in [-k,+k]^n}\left|p_y-q_y\right|\leq \beta$, which, in terms of $x$, because we are summing over all strings in the orbit under $(S_k)^n$ symmetry, can be written:
$$\frac{1}{2}\sum\limits_{x\in \{\pm 1\}^{kn}}\frac{\left|p_{\phi(x)}-q_{\phi(x)}\right|}{\orbitsize}\leq \beta$$\\

First we define for each $x$, $\Delta_x=\frac{|p_{\phi(x)}-q_{\phi(x)}|}{\orbitsize}$ and so $\|\disttwo-\mathcal{D}'\|=\frac{1}{2}\sum\limits_x \Delta_x$. \\

Note that:

$$\E_x[\Delta_x]=\frac{\sum\limits_x \Delta_x}{2^{kn}}=\frac{2\beta}{2^{kn}}$$

And applying Markov, $\forall j>1$,

$$\Pr_x[\Delta_x>\frac{j2\beta}{2^{kn}}]<\frac{1}{j}$$

Setting $j=\frac{4}{\delta}, \beta=\frac{\epsilon\delta}{16}$, we have,

$$\Pr_x[\Delta_x>\frac{\epsilon}{2}\cdot\frac{1}{2^{kn}}]<\frac{\delta}{4}$$

Then use approximate counting (with an $\NP$ oracle), using Theorem \ref{stockmeyer} on the randomness of $S$ to obtain an output $\tilde{q}_y$ so that, for all $\gamma>0$, in time polynomial in $n$ and $\frac{1}{\gamma}$:

$$\Pr[|\tilde{q}_y-q_y| >\gamma\cdot q_y]<\frac{1}{2^{n}}$$
Because we can amplify the failure probability of Stockmeyer's algorithm to be inverse exponential.\\

Equivalently in terms of $x$:
$$\Pr_x\left[\frac{|\tilde{q}_{\phi(x)}-q_{\phi(x)}|}{\orbitsize} >\gamma\cdot \frac{q_{\phi(x)}}{\orbitsize }\right]<\frac{1}{2^{n}}$$

And we have:

$$\E_x\left[\frac{q_{\phi(x)}}{\orbitsize}\right]\leq\frac{\sum\limits_x\frac{q_{\phi(x)}}{{\orbitsize}}}{2^{kn}}=\frac{1}{2^{kn}}$$

Thus, by Markov,

$$\Pr_x[\frac{q_{\phi(x)}}{{{\orbitsize}}}>\frac{j}{2^{kn}}]<\frac{1}{j}$$\\

Now, setting $\gamma=\frac{\epsilon\delta}{8}$ and applying the union bound:

\begin{align*}
\Pr_x & \left[\frac{\left|\tilde{q}_{\phi(x)}-p_{\phi(x)}\right|}{{{\orbitsize}}}>\frac{\epsilon}{2^{kn}}\right]\\&\leq \Pr_x\left[\frac{\left|\tilde{q}_{\phi(x)}-q_{\phi(x)}\right|}{\orbitsize}>\frac{\epsilon}{2}\cdot\frac{1}{2^{kn}}\right]+\Pr_x\left[\frac{\left|q_{\phi(x)}-p_{\phi(x)}\right|}{\orbitsize}>\frac{\epsilon}{2}\cdot\frac{1}{2^{kn}}\right]
\\  & \leq \Pr_x\left[\frac{q_{\phi(x)}}{\orbitsize}>\frac{j}{2^{kn}}\right]\\&+\Pr\left[\frac{|\tilde{q}_{\phi(x)}-q_{\phi(x)}|}{\orbitsize}>\gamma\cdot \frac{q_{\phi(x)}}{\orbitsize}\right]+\Pr_x\left[\Delta_x>\frac{\epsilon}{2}\cdot\frac{1}{2^{kn}}\right]
\\ & \leq\frac{1}{j}+\frac{1}{2^{n}}+\frac{\delta}{4}
\\ & \leq \frac{\delta}{2}+\frac{1}{2^{n}}\leq \delta.
\end{align*}

\end{proof}
\section{Putting it All Together}\label{sec-summary}
In this section we put our results in perspective and conclude.

As mentioned before, our goal is to find a class of distributions $\{\mathcal{D}_{n}\}_{\rangeovern}$ that can be sampled exactly in $\poly(n)$ time on a Quantum Computer, with the property that there does not exist a (classical) Sampler relative to that class of distributions, $\{\mathcal{D}_n\}_{\rangeovern}$.  

Using the results in Sections \ref{sec-classical-nonsquash} and \ref{sec-non-squash} we can quantumly sample from a class of distributions $\{\disttwo\}_{\rangeovern} $, where $k\in\poly(n)$ with the property that, if there exists a classical Sampler relative to this class of distributions, there exists an $\epsilon\var$-additive $\delta$-average case solution to the $Q^2$ function with respect to the $\mathcal{B}(0,k)^n$ distribution.  If we had an efficient decomposition for the ``Squashed QFT" unitary matrix, we could use the results from Sections \ref{sec-quantumsquash} and \ref{sec-classicalsquash} to make $k$ as large as $\exp(n)$.    We would like this to be an infeasible proposition, and so we conjecture:

\begin{conjecture} There exists some Efficiently Specifiable polynomial $Q$ with $n$ variables, so that $\epsilon\var$-additive $\delta$-average case solutions with respect to $\mathcal{B}(0,k)^n$, for any fixed $k<exp(n)$, to $Q^2$, cannot be computed in (classical) randomized $\poly(n,1/\epsilon,1/\delta)$ time with a $\PH$ oracle.
\end{conjecture}

At the moment we don't know of such a decomposition for the ``Squashed QFT".  However, we do know that we can classically evaluate a related fast (time $n\log^2{n}$)  polynomial transform by a theorem of Driscoll, Healy, and Rockmore \cite{dhr}.  We wonder if there is some way to use intuition gained by the existence of this fast polynomial transform to show the existence of an efficient decomposition for our ``Squashed QFT".

Additionally, if we can prove the Anti-Concentration Conjecture (Conjecture \ref{anticon}) relative to some Efficiently Specifiable polynomial $Q$ and the $\mathcal{B}(0,k)^n$ distribution, we appeal to Theorem \ref{thm-anticon} to show that it suffices to prove:

\begin{conjecture} There exists some Efficiently Specifiable polynomial $Q$ with $n$ variables, so that $Q$ satisfies Conjecture \ref{anticon} relative to $\mathcal{B}(0,k)^n$, for $k\leq \exp(n)$, and $\epsilon$-multiplicative $\delta$-average case solutions, with respect to $\mathcal{B}(0,k)^n$, to $Q^2$ cannot be computed in (classical) randomized $\poly(n,1/\epsilon,1/\delta)$ time with a $\PH$ oracle.
\end{conjecture}

We would be happy to prove that either of these two solutions (additive or multiplicative) are $\sharpP$-hard.  In this case we can simply invoke Toda's Theorem \cite{toda} to show that such a randomized classical solution would collapse $\PH$ to some finite level.

We note that at present, both of these conjectures seem out of reach, because we do not have an example of a polynomial that is $\sharpP$-hard to approximate (in either multiplicative or additive) on average, in the sense that we need.  Hopefully this is a consequence of a failure of proof techniques, and can be addressed in the future with new ideas.
\bibliographystyle{alpha}

\bibliography{fefferman}
\appendix

\section{The Power of Exact Quantum Sampling}\label{chapter-exact}
For the sake of completeness, in this section we prove a folklore result (that is implicit in e.g., \cite{ScottSharp}) showing that, unless the $\PH$ collapses to a finite level, there is a class of distributions that can be sampled efficiently on a Quantum Computer, that cannot be sampled exactly classically.  

Note that as a consequence of Theorem \ref{stockmeyer}, given an efficiently computable $f:\{0,1\}^n\rightarrow\{0,1\}$ we can compute a multiplicative approximation to $\Pr_{x\sim U_{\{0,1\}^n}}[f(x)=1]=\frac{\sum\limits_{x\in\{0,1\}^n}f(x)}{2^{n}}$ in the $\PH$.

Now we show the promised class of quantumly sampleable distributions:

\begin{definition}[$\mathcal{D}_{f}$]
Given $f:\{0,1\}^n\rightarrow\{\pm 1\}$, we define the distribution $\mathcal{D}_{f}$ over $\{0,1\}^n$ as follows:
$$\Pr_{\mathcal{D}_{f,n}}\left[y\right]=\frac{\left(\sum\limits_{x\in\{0,1\}^n}(-1)^{\langle x,y\rangle}f(x)\right)^2}{2^{2n}}$$
\end{definition}
The fact that this is a distribution will follow from the proceeding discussion.
\begin{theorem}
For all efficiently computable $f:\{0,1\}^n\rightarrow\{\pm 1\}$ we can sample from $\mathcal{D}_{f}$ in $\poly(n)$ time on a Quantum Computer.
\end{theorem}
\begin{proof}
Consider the following quantum algorithm:
\begin{enumerate}
\item{Prepare the uniform superposition over $n$ qubits, $\frac{1}{2^{n/2}}\sum\limits_{x\in\{0,1\}^n}|x\rangle$}
\item{Since by assumption $f$ is efficiently computable, we can apply $f$ to the phases (as discussed in Section \ref{sec:preliminaries}), with two quantum queries to $f$ resulting in: $$|f\rangle=\frac{1}{2^{n/2}}\sum\limits_{x\in\{0,1\}^n}f(x)|x\rangle$$}
\item{Apply the $n$ qubit Hadamard, $H^{\otimes n}$}
\item{Measure in the standard basis}
\end{enumerate}
Note that $H^{\otimes n}|f\rangle=\frac{1}{2^n}\sum\limits_{y\in\{0,1\}^n}\sum\limits_{x\in\{0,1\}^n}(-1)^{\langle x,y\rangle}f(x)|y\rangle$ and therefore the distribution sampled by the above quantum algorithm is $\mathcal{D}_{f}$.  
\end{proof}
As before, the key observation is that $\left(\langle 00...0|H^{\otimes n}|f\rangle\right)^2=\frac{\left(\sum\limits_{x\in\{0,1\}^n}f(x)\right)^2}{2^{2n}}$, and therefore encodes a $\sharpP$-hard quantity in an exponentially small amplitude.  We can exploit this hardness classically if we assume the existence of a classical sampler, which we define to mean an efficient random algorithm whose output is distributed via this distribution.

\begin{theorem}[Folklore, e.g., \cite{ScottSharp}]\label{exact}
Suppose we have a classical randomized algorithm $B$, which given as input $0^n$, samples from $\mathcal{D}_{f}$ in time $\poly(n)$, then the $\PH$ collapses to $\BPP^{\NP}$.
\end{theorem}
\begin{proof}

The proof follows by applying Theorem \ref{stockmeyer} to obtain an approximate count to the fraction of random strings $r$ so that $B(0^n,r)=00..0$.  Formally, we can output an $\alpha$ so that:
$$\left(1-\epsilon\right)\frac{\left(\sum\limits_{x\in\{0,1\}^n}f(x)\right)^2}{2^{2n}}\leq \alpha \leq \frac{\left(\sum\limits_{x\in\{0,1\}^n}f(x)\right)^2}{2^{2n}}\left(1+\epsilon\right)$$ In time $\poly(n,1/\epsilon)$ using an $\NP$ oracle.  Multiplying through by $2^{2n}$ allows us to get a multiplicative approximation to $\left(\sum\limits_{x\in\{0,1\}^n}f(x)\right)^2$ in the $\PH$.  It is clear that, given efficiently computable $f:\{0,1\}^n\rightarrow\{\pm 1\}$  computing $\sum\limits_{x\in\{0,1\}^n} f(x)$ is $\sharpP$-hard. Aaronson \cite{ScottSharp} has shown that even calculating this relative error estimate to $\left(\sum\limits_{x\in\{0,1\}^n} f(x)\right)^2$ is $\sharpP$-hard.  Since we know by Toda's Theorem \cite{toda}, $\PH\subseteq\Poly^{\sharpP}$, we now have that $\Poly^{\sharpP}\subseteq\BPP^{\NP}\Rightarrow\PH\subseteq\BPP^{\NP}$ leading to our theorem.  Note also that this theorem would hold even under the weaker assumption that the sampler is contained in $\BPP^{\PH}$.
\end{proof}
We end this Section by noting that Theorem \ref{exact} is extremely sensitive to the exactness condition imposed on the classical sampler, because the amplitude of the quantum state on which we based our hardness is only exponentially small.  Thus it is clear that by weakening our sampler to an ``approximate" setting in which the sampler is free to sample any distribution $Y$ so that the Total Variation distance $\|Y-\mathcal{D}_{f}\|\leq {1/\poly(n)}$ we no longer can guarantee any complexity consequence using the above construction.  Indeed, this observation makes the construction quite weak-- for instance, it may even be unfair to demand that any physical realization of this quantum circuit {\sl itself} samples exactly from this distribution!  In the preceding sections we are motivated by this apparent weakness and discuss the intractability of approximately sampling in this manner from quantumly sampleable distributions.

\section{Computation of the Variance of Efficiently Specifiable Polynomial}\label{variance-section}
In this section we compute the variance of the distribution  over $\mathbb{R}$ induced by an Efficiently Specifiable polynomial $Q$ with assignments to the variables chosen independently from the $\mathcal{B}(0,k)$ distribution.  We will denote this throughout the section by $\var$.  Recall, by the definition of Efficiently Specifiable, we have that $Q$ is an $n$ variate homogenous multilinear polynomial with $\{0,1\}$ coefficients.  Assume $Q$ is of degree $d$ and has $m$ monomials.  Let each $[-k,k]$ valued variable $X_i$ be independently distributed from $\mathcal{B}(0,k)$.

We adopt the notation whereby, for $j\in[m],l \in [d]$, $x_{j_l}$ is the $l$-th variable in the $j$-th monomial of $Q$.  

Using the notation we can express $Q(X_1,...,X_n)=\sum\limits_{j=1}^{m}\prod\limits_{l=1}^d X_{j_l}$.  By independence of these random variables and since they are mean $0$, it suffices to compute the variance of each monomial and multiply by $m$:

\begin{align}
\var=\Var\left[Q(X_1,...,X_n)\right] &=\E\left[\sum\limits_{j=1}^{m}\prod\limits_{l=1}^d X_{j_l}^2\right]=\sum\limits_{j=1}^{m}\E\left[\prod\limits_{l=1}^d X_{j_l}^2\right]\\
&=m\E\left[\prod\limits_{l=1}^d X_{1_l}^2\right]=m\prod\limits_{l=1}^d\E\left[X_{1_l}^2\right]\\&=m\left(\E\left[X_{1_1}^2\right]\right)^d
\end{align}

Now since these random variables are independent and identically distributed, we can calculate the variance of an arbitrary $X_{j_l}$ for any $j\in[m]$ and $l \in [d]$:
\begin{align}
\E[X_{j_l}^2] &= \frac{1}{2^{k}}\sum\limits_{i=0}^k\left[\left(k-2i\right)^2\binom{k}{i}\right]\\
\end{align}

Thus, the variance of $Q$ is: 
$$m\frac{1}{2^{kd}}\left(\sum\limits_{i=0}^k\left[\left(k-2i\right)^2\binom{k}{i}\right]\right)^d$$

It will be useful to calculate this variance in a different way, and obtain a simple closed form.  In this way we will consider the $k$-valued equivalent polynomial $Q'_k:\mathbb{T}_2^{nk}\rightarrow\mathbb{R}$ which is a sum of $m'=mk^d$ multilinear monomials, each of degree $d$.  As before we can write $Q'_k(X_1,...,X_{nk})=\sum\limits_{j=1}^{m'}\prod\limits_{l=1}^d X_{j_l}$.  Note that the uniform distribution over assignments in $\mathbb{T}_2^{kn}$ to $Q'_k$ induces $\mathcal{B}(0,k)^n$ over $[-k,k]^n$ assignments to $Q$.  By the same argument as above, using symmetry and independence of random variables, we have:

\begin{align}
\var=\Var[Q(X_1,X_2,...,X_{n})]=&
\Var[Q'_k(X_1,X_2,...,X_{nk})]\\&=m'\prod\limits_{l=1}^{d}\E\left[X_{1_l}^2\right]
\\&=m'\E\left[X_{1_1}^2\right]^d=1^d m'=m'=k^dm
\end{align}

\section{Examples of Efficiently Specifiable Polynomials}\label{chap-efficientlyspecifiable}
In this section we give two examples of Efficiently Specifiable polynomials.

\begin{theorem}
${\bf Permanent}\left(x_1,...,x_{n^2}\right)=\sum\limits_{\sigma\in S_{n}}\prod\limits_{i=1}^{n}x_{i,\sigma{\left(i\right)}}$ is Efficiently Specifiable.
\end{theorem}

\begin{proof}
 We note that it will be convenient in this section to index starting from $0$.  The theorem follows from the existence of an $h_{\bf Permanent}:[0,n!-1]\rightarrow\{ 0,1\}^{n^2}$ that efficiently maps the $i$-th permutation over $n$ elements to a string representing its obvious encoding as an $n\times n$ permutation matrix.  We will prove that such an efficiently computable $h_{\bf Permanent}$ exists and prove that its inverse, $h^{-1}_{\bf Permanent}$ is also efficiently computable. \\

The existence of $h_{\bf Permanent}$ follows from the so-called ``factorial number system" \cite{knuth}, which gives an efficient bijection that associates each number in $[0,n!-1]$ with a permutation in $S_n$.  It is customary to think of the permutation encoded by the factorial number system as a permuted sequence of $n$ numbers, so that each permutation is encoded in $n\log{n}$ bits.  However, it is clear that we can efficiently transform this notation into the $n \times n$ permutation matrix.

To go from an integer $j\in[0,n!-1]$ to its permutation we:
\begin{enumerate}
\item{Take $j$ to its ``factorial representation", an $n$ number sequence, where the $i$-th place value is associated with $(i-1)!$, and the sum of the digits multiplied by the respective place value is the value of the number itself.  We achieve this representation by starting from $(n-1)!$, setting the leftmost value of the representation to $j'=\lfloor\frac{j}{(n-1)!}\rfloor$, letting the next value be $\lfloor\frac{j-j'\cdot(n-1)!}{(n-2)!}\rfloor$ and continuing until $0$.  Clearly this process can be efficiently achieved and efficiently inverted, and observe that the largest each value in the $i$-th place value can be is $i$. }
\item{In each step we maintain a list $\ell$ which we think of as originally containing $n$ numbers in ascending order from $0$ to $n-1$.}
\item{Repeat this step $n$ times, once for each number in the factorial representation.  Going from left to right, start with the left-most number in the representation and output the value in that position in the list, $\ell$.  Remove that position from $\ell$.}
\item{The resulting $n$ number sequence is the encoding of the permutation, in the standard $n\log{n}$ bit encoding}
\end{enumerate}
 \end{proof}
Now we show that the Hamiltonian Cycle Polynomial is Efficiently Specifiable.

Given a graph $G$ on $n$ vertices, we say a Hamiltonian Cycle is a path in $G$ that starts at a given vertex, visits each vertex in the graph exactly once and returns to the start vertex.  

We define an $n$-cycle to be a Hamiltonian cycle in the complete graph on $n$ vertices.  Note that there are exactly $(n-1)!$ $n$-cycles in $S_n$.
  
\begin{theorem} ${\bf HamiltonianCycle}\left(x_1,...,x_{n^2}\right)=\sum\limits_{\sigma: \:n-cycle}\prod\limits_{i=1}^{n}x_{i,\sigma{\left(i\right)}}$ is Efficiently Specifiable.
\end{theorem}
\begin{proof}
We can modify the algorithm for the Permanent above to give us an efficiently computable $h_{HC}:[0,(n-1)!-1]\rightarrow\{0,1\}^{n^2}$ with an efficiently computable $h_{HC}^{-1}$.

To go from a number $j\in[0,(n-1)!-1]$ to its $n$-cycle we:
\begin{enumerate}
\item{Take $j$ to its factorial representation as above.  Now this is an $n-1$ number sequence where the $i$-th place value is associated with $(i-1)!$, and the sum of the digits multiplied by the respective place value is the value of the number itself.}
\item{In each step we maintain a list $\ell$ which we think of as originally containing $n$ numbers in ascending order from $0$ to $n-1$.}
\item{Repeat this step $n-1$ times, once for each number in the factorial representation.  First remove the smallest element of the list.  Then going from left to right, start with the left-most number in the representation and output the value in that position in the list, $\ell$.  Remove that position from $\ell$.}
\item{We output $0$ as the $n$-th value of our $n$-cycle.}
\end{enumerate}
To take an $n$-cycle to a factorial representation, we can easily invert the process:  
\begin{enumerate}
\item{In each step we maintain a list $\ell$ which we think of as originally containing $n$ numbers in order from $0$ to $n-1$.}
\item{Repeat this step $n-1$ times.  Remove the smallest element of the list.  Going from left to right, start with the left-most number in the $n$-cycle and output the position of that number in the list $\ell$ (where we index the list starting with the $0$ position).  Remove the number at this position from $\ell$.}
\end{enumerate}

\end{proof}

\section{A Simple Example of ``Squashed" QFT, for $k=2$}
In this Section we explicitly construct the matrix $L\cdot \tilde{D}_2 \cdot R$ from the $QFT$ over $\mathbb{Z}_2^2$.  Note that the matrix we referred to as $D_2$ is:

\[ \left( \begin{array}{cccc}
1 & 1 & 1 & 1\\
1 & -1 & 1 & -1\\
1 & 1 & -1 & -1\\
1 & -1 & -1 & 1\end{array} \right)\]

Where we can think of the columns as identified with the monomials $\{1, x_1, x_2, x_1x_2\}$ in this order (from left to right) and the rows (from top to bottom) as identified with the assignments $\{(1,1),(-1,1),(1,-1),(-1,-1)\}$ where the first element in each pair is the assignment to $x_1$ and the second is to $x_2$.  Note that as desired, the $(i,j)$-th element of $D_2$ is the evaluation of the $j$-th monomial on the $i$-th assignment.\\

Now we create $\tilde{D}^{(1)}_2$ by combining columns of monomials that belong to each elementary symmetric polynomial, as described in the prior section.  We identify the columns with elementary symmetric polynomials on variables $x_1,x_2$ in order from left to right: $1,x_1+x_2,x_1x_2$ and the rows remain the same.  This gives us:

\[ \left( \begin{array}{ccc}
1 & 2 & 1\\
1 & 0 & -1\\
1 & 0 & -1\\
1 & -2 & 1\end{array} \right)\]

It can easily be verified that the columns are still orthogonal.  Now we note that the rows corresponding to assignments $(1,-1)$ and $(-1,1)$ are in the same orbit with respect to $S_2$ symmetry.  And thus we obtain $\tilde{D}_2$:
\[ \left( \begin{array}{ccc}
1 & 2 & 1\\
1 & 0 & -1\\
1 & -2 & 1\end{array} \right)\]

Now $L$ is the diagonal matrix whose $i$-th entry is $\sqrt{o_i}$, the size of the $i$-th equivalence class of assignments under $S_2$ symmetry.  Note that $|o_0|=\sqrt{\binom{2}{0}}=1$, $|o_1|=\sqrt{\binom{2}{1}}=\sqrt{2}$, and $|o_2|=\sqrt{\binom{2}{2}}=1$, and so $L$ is:
\[ \left( \begin{array}{ccc}
1 & 0 & 0\\
0 & \sqrt{2} & 0\\
0 & 0 & 1\end{array} \right)\]

And $L\cdot \tilde{D}_2=$
\[ \left( \begin{array}{ccc}
1 & 2 & 1\\
\sqrt{2} & 0 & -\sqrt{2}\\
1 & -2& 1\end{array} \right)\]

And we note that the columns are now orthogonal.  As before, this implies there exists a diagonal matrix $R$ so that $L\cdot \tilde{D}_2\cdot R$ is unitary.  It is easily verified that this is the matrix $R$:

\[ \left( \begin{array}{ccc}
\frac{1}{2} & 0 & 0\\
0 & \frac{1}{\sqrt{8}} & 0\\
0 & 0& \frac{1}{2}\end{array} \right)\]

And the first two elements $r_0,r_1$ can be easily seen to be $\frac{1}{\sqrt{2^k}}=\frac{1}{2}$ and $\frac{1}{\sqrt{\sum\limits_{i=0}^k\left[\binom{k}{i}\left(k-2i\right)^2\right]}}=\frac{1}{\sqrt{8}}$, as claimed in the prior section.  Thus the final $k+1 \times k+1$ matrix $L\cdot \tilde{D}_2\cdot R$ is:

\[ \left( \begin{array}{ccc}
\frac{1}{2} & \frac{2}{\sqrt{8}} & \frac{1}{2} \\
\frac{\sqrt{2}}{2} & 0 & -\frac{\sqrt{2}}{2}\\
\frac{1}{2} & -\frac{2}{\sqrt{8}}& \frac{1}{2}\end{array} \right)\]

Which is unitary, as desired.

\end{document}